\documentclass[11pt, a4paper]{article}

\usepackage{ifthen}

\makeatletter
{}

{}
\makeatother

\newcommand{\ifndef}[2]{\ifthenelse{\isundefined{#1}}{#2}{}}

\newcommand{\mydef}[2]{\def#1{#2}}

\newcommand{\nospell}[1]{#1}   %

\makeatletter
\@ifpackageloaded{amssymb}{}{\usepackage{amssymb}}
\@ifpackageloaded{amsmath}{}{\usepackage{amsmath}}   %
\@ifpackageloaded{amsthm}{}{      %
\usepackage{amsthm}
}
\makeatother

{}
\usepackage{latexsym}
\usepackage{amsfonts}
\usepackage{units}     %
\usepackage{psfrag}    %
\usepackage{url}    %
\usepackage{hyphenat}   %
{}
{}
\usepackage[dvips]{graphicx}
\usepackage[usenames,dvipsnames]{color}
{}
{}
{}

\ifndef{\theorem}{\newtheorem{theorem}{Theorem}[section]}
\ifndef{\lemma}{\newtheorem{lemma}[theorem]{Lemma}}
\ifndef{\corollary}{\newtheorem{corollary}[theorem]{Corollary}}
\ifndef{\conjecture}{\newtheorem{conjecture}[theorem]{Conjecture}}
\ifndef{\remark}{\theoremstyle{remark} \newtheorem{remark}{Remark}}
\ifndef{\proposition}{\newtheorem{proposition}{Proposition}}
\ifndef{\claim}{\newtheorem{claim}[theorem]{Claim}}
\ifndef{\result}{\newtheorem{result}[theorem]{Result}}
\ifndef{\problem}{\newtheorem{problem}{Problem}}

\newtheoremstyle{mydefinition}   %
{\topsep}{\topsep}   %
{\slshape}   %
{}   %
{\bfseries}   %
{.}   %
{ }   %
{}   %

\newtheoremstyle{myremark}   %
{\topsep}{\topsep}   %
{\slshape}   %
{}   %
{\bfseries\slshape}   %
{:}   %
{ }   %
{}   %

\newtheoremstyle{myexample}   %
{\topsep}{\topsep}   %
{\itshape}   %
{}   %
{\slshape}   %
{:}   %
{ }   %
{\ul{\thmname{#1}}}   %

\newtheoremstyle{myclaims}   %
{\topsep}{\topsep}   %
{\slshape}   %
{}   %
{\bfseries\itshape}   %
{}   %
{ }   %
{\thmname{#1}\thmnumber{ \!#2}.}   %

{}  %
{}  %
{\theoremstyle{mydefinition}}
{\theoremstyle{myremark}}
\ifndef{\example}{\theoremstyle{myexample} \newtheorem{example}{Example}}
{\theoremstyle{myclaims}

\ifndef{\fact}{\newtheorem{fact}[theorem]{Fact}}
}
{}  %
{}  %

{}
\newtheoremstyle{anystatement}{\topsep}{\topsep}{\itshape}{}{\bfseries}{.}{ }{\anystatementname}
{\theoremstyle{anystatement}}
\newcommand{\anystatementname}{}

{}

\newcommand{\AuxNew}[4][]{#2{#3}[1][*]%
{\ifthenelse{\equal{*}{##1}} %
{\Ensuremath{#1{#4}}}%
{\ifthenelse{\equal{b}{##1}} %
{\Ensuremath{\mathbf{#4}}}%
{\ifthenelse{\equal{}{##1}} %
{\IfMathMode{#1{#4}}{#4}}{}}}}}

\newcommand{\newident}[3][*]{\ifthenelse{\equal{*}{#1}}%
{\AuxNew[\mathit]{\newcommand}{#2}{#3}} %
{\mydef{#2}{\Ensuremath{\mathit{#3}}}}} %

\newcommand{\newmat}[3][*]{\ifthenelse{\equal{*}{#1}}%
{\AuxNew{\newcommand}{#2}{#3}} %
{\mydef{#2}{\Ensuremath{#3}}}} %

\newcommand{\providemat}[3][*]{\ifthenelse{\equal{*}{#1}} %
{\AuxNew{\providecommand}{#2}{#3}} %
{\mydef{#2}{\Ensuremath{#3}}}} %

\newcommand{\providematarg}[2]{ %
\providecommand{#1}[1][]{\Ensuremath{#2}}}

\newcommand{\newfunction}[2]{ %
\newcommand{#1}[2][*]{\ifthenelse{\equal{*}{##1}}%
{\Ensuremath{#2{\left(##2\right)}}}%
{#2(##2)}}%
}

\newcommand{\MyMakeTheoMacros}[3]{
\newcommand{#2}[2][]{\ifthenelse{\equal{}{##1}}
{\begin{#1} ##2 \end{#1}}
{\begin{#1}\label{##1} ##2\end{#1}}}
\newcommand{#3}[3][]{\ifthenelse{\equal{}{##1}}
{\begin{#1}[##2] ##3 \end{#1}}
{\begin{#1}[##2]\label{##1} ##3\end{#1}}}
}

\makeatletter
\newtheorem*{rep@theorem}{\rep@title}
\newcommand{\newreptheorem}[2]{%
\newenvironment{rep#1}[1]{%
\def\rep@title{#2 \ref{##1}}%
\begin{rep@theorem}}%
{\end{rep@theorem}}}
\makeatother

\newcommand{\MyMakeDupTheoMacros}[7]{
\MyMakeTheoMacros{#1}{#2}{#3}
\newreptheorem{#1}{#6}
\newcommand{#4}[3]{
\newcommand{##2}{##3}
\begin{#1}\label{##1} ##2\end{#1}}
\newcommand{#5}[4]{
\newcommand{##2}{##4}
\begin{#1}{\e{##3}}\label{##1} ##2\end{#1}}
\newcommand{#7}[2]{\begin{rep#1}{##1} ##2 \end{rep#1}}
}

\newcommand{\MyMakeRefMacros}[3]{\newcommand{#1}[2][]
{\ifthenelse{\equal{}{##1}}{#2~\ref{##2}}{#3~\ref{##1} and~\ref{##2}}}}

\newcommand{\MyMakeEqRefMacros}[3]{\newcommand{#1}[2][]
{\ifthenelse{\equal{}{##1}}{#2~\eqref{##2}}{#3~\eqref{##1} and~\eqref{##2}}}}

{}   %
\newcommand{\bibentry}[8]{
{}\bibitem[\nospell{#8}]{#1} {\textup #3}.{}
\ifthenelse{\equal{}{#6}}
{\newblock \textrm{#4.} \newblock {\em #5}, #7.}
{\newblock \textrm{#4.} \newblock {\em #5, #6}, #7.}
}

{}  %

{}  %
\MyMakeDupTheoMacros{lemma}
{\lem}{\nlem}{\lemdup}{\nlemdup}{Lemma}{\lemrep}
{}

\MyMakeRefMacros{\lemref}{Lemma}{Lemmas}

\MyMakeDupTheoMacros{corollary}
{\crl}{\ncrl}{\crldup}{\ncrldup}{Corollary}{\crlrep}

\MyMakeRefMacros{\crlref}{Corollary}{Corollaries}

\MyMakeTheoMacros{proposition}{\prp}{\nprp}

{}
\newtheorem*{prp*}{\e{Proposition}}

{}

\MyMakeRefMacros{\prpref}{Proposition}{Propositions}

\MyMakeDupTheoMacros{my_claim}
{\clm}{\nclm}{\clmdup}{\nclmdup}{Claim}{\clmrep}

\MyMakeRefMacros{\clmref}{Claim}{Claims}

\MyMakeDupTheoMacros{theorem}
{\theo}{\ntheo}{\theodup}{\ntheodup}{Theorem}{\theorep}

\MyMakeRefMacros{\theoref}{Theorem}{Theorems}

\MyMakeTheoMacros{definition}{\defi}{\ndefi}

\MyMakeRefMacros{\defiref}{Definition}{Definitions}

\MyMakeTheoMacros{problem}{\prob}{\nprob}

\MyMakeRefMacros{\probref}{Problem}{Problems}

\MyMakeTheoMacros{conjecture}{\conj}{\nconj}

\MyMakeRefMacros{\conjref}{Conjecture}{Conjectures}

\newcommand{\fakeconjref}[1]{Conjecture~{#1}}

\providecommand{\qedsymbol}{\square}

\newcommand{\prf}[2][]{\ifthenelse{\equal{}{#1}}%
{\begin{proof}\renewcommand{\qedsymbol}{$\blacksquare$}%
#2 \end{proof}}%
{\begin{proof}[Proof of #1]%
\renewcommand{\qedsymbol}{$\blacksquare_{\mbox{\it{\scriptsize{#1}}}}$}%
#2 \end{proof}}
}

\newcommand{\abstr}[1]{\begin{abstract} #1 \end{abstract}}

\newcommand{\sect}[2][]{\ifthenelse{\equal{}{#1}}
{\section{#2}}
{\section{#2}\label{#1}}}

\newcommand{\ssect}[2][]{\ifthenelse{\equal{}{#1}}
{\subsection{#2}}
{\subsection{#2}\label{#1}}}

\MyMakeRefMacros{\chref}{Chapter}{Chapters}

\MyMakeRefMacros{\sref}{Section}{Sections}

\MyMakeRefMacros{\ssref}{Subsection}{Subsections}

\MyMakeRefMacros{\sssref}{Subsection}{Subsections}

\MyMakeRefMacros{\figref}{Figure}{Figures}

\newcommand{\IfMathMode}[2]{\ifmmode{#1}\else{#2}\fi}

\newcommand{\Ensuremath}{\ensuremath}

\newcommand{\fbr}[1]{\IfMathMode %
{#1}{$#1$}}                      %

\newcommand{\fnbr}[1]{\mbox{\fbr{#1}}}   %

\newcommand{\fla}[2][*]{\ifthenelse{\equal{}{#1}}{\fbr{#2}}{\fnbr{#2}}}

\newcommand{\mat}[2][]{\ifthenelse{\equal{}{#1}} %
{ \begin{displaymath} #2 \end{displaymath} } %
{ \begin{equation} \label{#1} #2 \end{equation} }%
}

\newcommand{\matal}[2][]{\mat[#1]{\begin{aligned} #2 \end{aligned}}}

\newcommand{\f}{\fla}

\newcommand{\m}{\mat}
\newcommand{\mal}{\matal}

\MyMakeEqRefMacros{\equref}{Equation}{Equations}

\MyMakeEqRefMacros{\expref}{Expression}{Expressions}

\MyMakeEqRefMacros{\inequref}{Inequality}{Inequalities}

\newcommand{\bracref}[1]{(\ref{#1})}

\newcommand{\bref}{\bracref}

\newcommand{\twocase}[4]%
{\begin{cases}#1 &\txt{#2}\\#3 &\txt{#4}\end{cases}}

\providecommand{\middle}{\big}

\newcommand{\chs}{\genfrac(){0cm}{}}   %

\newcommand{\Hh}[2][]{\ifthenelse{\equal{}{#2}}%
{H_{#1}}%
{H_{#1}{\left[{#2}\right]}}}

\providecommand{\E}[2][]{\ifthenelse{\equal{}{#1}}%
{\mathop{\mathbf{E}}{\left[{#2}\right]}}%
{\mathop{\mathbf{E}}_{#1}{\left[{#2}\right]}}}

\newcommand{\Var}[1]{\mathop{\mathbf{Var}}{\left[{#1}\right]}}

\newcommand{\PR}[2][]{\mathop{\mathbf{Pr}}_{#1}{\left[{#2}\right]}}

\newcommand{\GF}[2][]{{\mathcal GF_{#2}^{#1}}}

\newcommand{\pl}[1][]{\nospell{\ifthenelse{\equal{}{#1}}%
{\!\stackrel'{}\!\!\txt{s}}%
{\fla{#1\!\stackrel'{}\!\!\txt{s}}}}}

\newcommand{\ord}[1][]{\nospell{\ifthenelse{\equal{}{#1}}%
{\txt{'th}}%
{\ifthenelse{\equal{1}{#1}}{$1\txt{'st}$}{\ifthenelse{\equal{2}{#1}}{$2\txt{'nd}$}{\ifthenelse{\equal{3}{#1}}{$3\txt{'rd}$}{\fla{#1\txt{'th}}}}}}}}

\newcommand{\fr}[3][*]{%
\ifthenelse{\equal{*}{#1}}        %
{\frac{#2}{#3}}{}%
\ifthenelse{\equal{/}{#1}}        %
{\nicefrac{#2}{#3}}{}%
\ifthenelse{\equal{}{#1}}         %
{\left.#2\middle/#3\right.}{}%
\ifthenelse{\equal{p_}{#1}}       %
{\left.\left(#2\right)\middle/#3\right.}{}%
\ifthenelse{\equal{_p}{#1}}       %
{\left.#2\middle/\left(#3\right)\right.}{}%
\ifthenelse{\equal{pp}{#1}}       %
{\left.\left(#2\right)\middle/\left(#3\right)\right.}{}
}

\newcommand{\dr}{\nicefrac}

\newcommand{\sq}{\mathpalette\MySQRT}  %

\def\MySQRT#1#2{    %
\setbox0=\hbox{$#1\sqrt{#2\,}$}\dimen0=\ht0%
\advance\dimen0-0.2\ht0%
\setbox2=\hbox{\vrule height\ht0 depth -\dimen0}%
{\box0\lower0.4pt\box2}}

\newcommand{\set}[2][]{\ifthenelse{\equal{}{#1}} %
{\Ensuremath{\left\{#2\right\}}}%
{\Ensuremath{\left\{#2\,\middle\arrowvert\,#1\right\}}}}

\newcommand{\Max}[2][]{\ifthenelse{\equal{}{#1}} %
{\Ensuremath{\max{\left\{#2\right\}}}}%
{\Ensuremath{\max{\left\{#2\,\middle\arrowvert\,#1\right\}}}}}

\providecommand{\ip}[2]{\Ensuremath{\lla #1,#2\rra}}

\mydef{\01}{\set{0,1}}

\renewcommand{\l}{\left}
\renewcommand{\r}{\right}

\providecommand{\norm}[2][]{\ifthenelse{\equal{}{#1}}%
{\Ensuremath{\left\|#2\right\|}}%
{\Ensuremath{\left\|#2\middle\|_{#1}\right.}}}

\newcommand{\txt}[1]{\textrm{#1}}   %

\DeclareMathAlphabet{\lowcal}{OT1}{pzc}{m}{it}

\newmat[]{\lla}{\left\langle}
\newmat[]{\rra}{\right\rangle}

\newmat[]{\tm}{\cdot}
\newmat[]{\sbseq}{\subseteq}
\newmat[]{\deq}{\stackrel{\textrm{def}}{=}}

\providemat{\QQ}{\mathbb{Q}}
\providematarg{\CC}{\ifthenelse{\equal{}{#1}}%
{\mathbb{C}}%
{\mathbb{C}^{#1}}}
\providematarg{\RR}{\ifthenelse{\equal{}{#1}}%
{\mathbb{R}}%
{\mathbb{R}^{#1}}}

\newcommand{\ds}[1][]
{\ifthenelse{\equal{}{#1}}{\dots}{#1\dots#1}}
\newmat[]{\dc}{\ds[,]}

\newcommand{\MyComment}[1]{\ClassWarning{My Macros}{#1}}

{}  %
\newcommand{\e}{\emph}
{}   %

\providecommand{\ul}[1]{\underline{#1}}  %

\newcommand{\tb}{\quad}

\newcommand{\iu}{{\rm i}}
\newcommand{\en}{{\rm e}}

\MyComment{Look for ...-s}

\MyComment{Spell-check}

{}
{}

\title{On the Joint Entropy of \f d-Wise-Independent Variables}

\date{\today}

\newcommand{\instDG}{Institute of Mathematics, Academy of Sciences, \v Zitna 25, Praha 1, Czech Republic.}
\newcommand{\thanksDG}{Partially funded by the grant P202/12/G061 of GA \v CR and by RVO:\ 67985840.}

{}
\author{Dmitry Gavinsky%
\thanks{\instDG\ \thanksDG}%
~\thanks{Part of this work was done while visiting the Centre for Quantum Technologies at the National University of Singapore, and was partially funded by the Singapore Ministry of Education and the NRF.}
\and Pavel Pudl\'ak\protect\footnotemark[1]
}
{}

\begin{document}

\maketitle

\thispagestyle{empty}

\abstr{How low can the joint entropy of $n$ \f d-wise independent (for $d\ge2$) discrete random variables be, subject to given constraints on the individual distributions (say, no value may be taken by a variable with probability greater than $p$, for $p<1$)?  
This question has been posed and partially answered in a recent work of Babai~\cite{B13_En}.

In this paper we improve some of his bounds, prove new bounds in a wider range of parameters and show matching upper bounds in some special cases.
In particular, we prove tight lower bounds for the min-entropy (as well as the entropy) of pairwise and three-wise independent balanced binary variables for infinitely many values of $n$.
}



\setcounter{page}{1}

\sect[s_intro]{Introduction}

Suitable choice of a (discrete) distribution is a crucial component
that underlies many results in extremal combinatorics and theoretical
computer sciences (e.g., see~\cite{AS08_The}).  It is often the case
that the ``ideal'' distribution to use would be mutually independent
over $n$ random variables $X_1\dc X_n$ (each variable taking one of
several possible values); however, ``full'' mutual independence is
``too expensive'' and a \f d-wise-independent distribution is used
instead (e.g., see~\cite{LW06_Pair}). (A string of random variables
$X_1,\dots,X_n$ is called $d$-wise independent if any $d$-tuple of the
variables is independent.) Indeed, if all variables are independent, then the sample space has at least exponential size, while $d$-wise independent spaces can be of polynomial size if $d$ is constant.
This has many applications in computer science. The
size of the space, the number of random bits needed and the joint
entropy of $X_1,\dots,X_n$ are closely related parameters that are
crucial in these applications. 

This is a motivation of the question
studied in a recent article of Babai~\cite{B13_En}: what is the
minimum entropy for $n$ pairwise independent variables. Babai showed an
asymptotically logarithmic lower bound, by proving a very nice
theorem. He proved that for any string $X_1,\dots,X_n$ of pairwise
independent binary-valued variables, where the probabilities are bounded away
from zero and one, there exists a logarithmic size subset of these
variables that is almost independent. Such a subset must have entropy
asymptotically equal to its size, so the logarithmic lower bound
follows.

Our aim in this paper is to answer some questions and improve bounds
of Babai.  For proving tight bounds, a more traditional approach (see
for example~\cite{L65_Pair}) based on a construction of orthogonal
matrices seems more suitable. This approach enables us, first, to extend
Babai's bounds to a larger range of parameters, and second, to obtain more
precise bounds. In particular, we prove that the joint entropy of  $X_1,\dots,X_n$ is logarithmic even if the entropy of the variables is only of the order of $\log n/n$, which is the lowest possible.
Furthermore, we prove a lower bound $\log(n+1)$, conjectured by Babai, on
the min-entropy of pairwise independent balanced binary variables
(i.e., when each $X_j$ is equal to $0$, respectively $1$, with
probability $1/2$). This matches the upper bounds given by the well
known construction based on Hadamard matrices. So the bound is tight
if an Hadamard matrix of dimension $n+1$ exists.

Lower bounds on the entropy of $d$-wise independent variables can be
obtained from lower bounds on pairwise independent variables by a
well-known construction that produces a longer string of pairwise
independent variables. We slightly modify this construction for odd values
of $d$ which enables us to obtain a matching upper and lower bounds
for $d=3$ and infinitely many values of $n$ (powers of $2$).

Although we are primarily interested in binary-valued variables, we will show that some of our lower bounds can be extended to the case of general (finite-outcome) pairwise independent variables.

\sect[s_prelim]{Preliminaries}

We will write $$\Hh X=\sum_x\PR{X=x}\tm\log\fr1{\PR{X=x}}$$ to denote the Shannon entropy of the (discrete) random variable $X$, and $$\Hh[min]X=\min_x\log\fr1{\PR{X=x}}$$ for the min-entropy.
Clearly, $\Hh[min]X\le\Hh X$.
All logarithms are to the base $2$.
Random variables $X_1\dc X_n$ are said to be \f d-wise independent if for every $s\in\chs{[n]}d$, the variables $(X_i)_{i\in s}$ are mutually independent.
A random variable is called \e{binary} if it is supported on $\01$.

Recall Cantelli's inequality~\cite{C10_Intorno} -- a strengthening of Chebyshev's inequality for the case of one-sided deviations:
\nlem[l_Cant]{Cantelli's inequality}{For every random variable $X$ and real $t>0$,
\begin{equation}\label{eq-2}
\PR{X\le\E X-t},~\PR{X\ge\E X+t}
\le\fr{1}{1+\fr{t^2}{\Var X}}.
\end{equation}}

\sect[s_low]{Lower bounds}

In this section we give lower bounds on the joint entropy of $n$ \f d-wise-independent variables.

\ssect[ss_Ber]{Pairwise independent binary variables}

Here we give two incomparable entropy lower bounds for the families of pairwise independent binary variables.

\theo[t_low]{Let $X=\l(X_1\dc X_n\r)$ be $n$ pairwise independent binary variables.
Let $q_j=\PR{X_j=1}$ and suppose that $0<q_j<1$ for $j=1\dc n$.
Then
\m{\Hh X\geq\sup_{0<t\leq n}\fr{\log (n+1-t)}{1+\fr 1{t^2}\sum_{j=1}^n\fr{(1-2q_j)^2}{q_j(1-q_j)}}.}
}

\prf{Let $A=\set{a_1\dc a_m}\sbseq\01^n$ be the support of $X$ and $p_i\deq\PR{X=a_i}$.
We will denote by $a_{ij}$ the \ord[j] element of $a_i$ for $j\in[n]$.

Define an $m\times (n+1)$ matrix
$U=\{u_{ij}\}$ as follows.
For all $i\in[m]$,
\m{u_{i0}\deq \sqrt{p_i},}
and for $j\in[n]$,
\m{u_{ij}\deq\twocase
{-\sq{\fr{p_iq_j}{1-q_j}}}{if $a_{ij}=0$}
{\sq{\fr{p_i(1-q_j)}{q_j}}}{if $a_{ij}=1$}~.}

For $0\le j\le n$, let $u_j$ denote the \ord[j] column vector of $U$; note that these vectors form an orthonormal family:
For $j>0$,
\m{\ip{u_0}{u_j}=
\sq{\fr{1-q_j}{q_j}}\tm\PR{X_j=1}-\sq{\fr{q_j}{1-q_j}}\tm\PR{X_j=0}=0;}
for $k>j>0$,
\mal{\ip{u_k}{u_j}
&=\sq{\fr{1-q_k}{q_k}}\tm
\l(\sq{\fr{1-q_j}{q_j}}\tm\PR{X_k=1\land X_j=1}
-\sq{\fr{q_j}{1-q_j}}\tm\PR{X_k=1\land X_j=0}\r)\\
&+\sq{\fr{q_k}{1-q_k}}\tm
\l(\sq{\fr{1-q_j}{q_j}}\tm\PR{X_k=0\land X_j=1}
-\sq{\fr{q_j}{1-q_j}}\tm\PR{X_k=0\land X_j=0}\r)\\
&=0,}
as follows from independence of $X_i$ and $X_j$.
As well, the norm of every $u_i$ is $1$.

Since the matrix $U$ is unitary, or can be made unitary by adding more
columns, we know that the norm of each row of $U$ is at most $1$.
Thus we get, for all $i\in[m]$,

\m[m_l1]{1\ge\sum_{j=0}^n u_{ij}^2
=p_i\tm\l(1+\sum_{j:a_{ij}=0}\fr{q_j}{1-q_j}+\sum_{j:a_{ij}=1}\fr{1-q_j}{q_j}\r)}
\[
=p_i\tm\left(1+\sum_{j=0}^n\left((1-a_{ij})\frac{q_j}{1-q_j}+a_{ij}\frac{1-q_j}{q_j}\right)\right).
\]

Our aim is to find a subset $A_0\subseteq A$ such that every string $a\in A_0$ has a low probability, whereas the weight of $A$ is large. This will give us the lower bound on the entropy.

Let $A_0$  be all the elements of $A$ satisfying
\m[e2]{1+\sum_{j=0}^n\left((1-a_{ij})\frac{q_j}{1-q_j}+a_{ij}\frac{1-q_j}{q_j}\right)
\ge n+1-t.}
W.l.o.g. we may assume that $A_0=\set{a_1\dc a_{m_0}}$.
Then, according to~(\ref{m_l1}), for every  $i=1,\dots,m_0$,
\begin{equation}\label{bound-p-i}
p_i\leq 1/(n+1-t).
\end{equation}
Let $Y$ be the random variable defined by
\[
Y:=1+\sum_{j=1}^n\left((1-X_{j})\frac{q_j}{1-q_j}+X_{j}\frac{1-q_j}{q_j}\right)
=1+\sum_{j=1}^n\frac{q_j}{1-q_j}+ \sum_{j=1}^n\frac{1-2q_j}{q_j(1-q_j)}\tm X_j.
\]
Then we have
\m{\sum_{i=1}^{m_0}p_j=\PR{Y\ge n+1-t}.}
The expectation of $Y$ is
\[
\E{Y}=1+\sum_{j=1}^n\left((1-q_{j})\frac{q_j}{1-q_j}+q_{j}\frac{1-q_j}{q_j}\right)=n+1.
\]
The variance of $Y$ is
\[
\Var{Y}=\Var{\sum_{j=1}^n\frac{1-2q_j}{q_j(1-q_j)}\tm X_j}=
\sum_{j=1}^n\Var{\frac{1-2q_j}{q_j(1-q_j)}\tm X_j}=
\]\[
\sum_{j=1}^n\left(\frac{1-2q_j}{q_j(1-q_j)}\right)^2\tm \Var{X_j}=
\sum_{j=1}^n\frac{(1-2q_j)^2}{q_j(1-q_j)},
\]
where we have used the fact that the variables are pairwise independent.
Now we apply Cantelli's inequality (\lemref{l_Cant}) to the random variable $Y$ and parameter $t$.
\[
\sum_{i=1}^{m_0}p_j=\PR{Y\ge n+1-t}
= \PR{Y\ge \E{Y}-t}
\geq 
\]\[
\frac 1{1+\frac 1{t^2}\Var{Y}}=
\frac 1{1+\frac 1{t^2}\sum_{j=1}^n{\frac{(1-2q_j)^2}{q_j(1-q_j)}}}.
\]
Using this inequality and the fact that $p_i^{-1}>n+1-t$ for all $i\in[m_0]$ (which is (\ref{bound-p-i})), we get
\m{\Hh X=\sum_{i=1}^{m}p_i\log p_i^{-1}
\ge\sum_{i=1}^{m_0}p_i\log p_i^{-1}
> \frac{\log(n+1-t)}{1+\frac 1{t^2}\sum_{j=1}^n{\frac{(1-2q_j)^2}{q_j(1-q_j)}}},}
as required.}

Suppose that $0<q\leq q_j\leq 1/2$ for some $q$ and all $j$. Since 
$\frac 1q\geq\frac{1-2q_j}{q_j(1-q_j)}$, we have
\begin{equation}\label{e4}
\Hh{X}\geq\sup_{0<t\leq n}\frac{\log(n+1-t)}{1+\frac{n}{t^2q}}.
\end{equation}
In particular, for $t=n/2$, 
\[
\Hh{X}\geq\frac{\log(n/2+1)}{1+\frac 4{nq}}.
\]
This proves that the entropy of $X$ is $\Omega(\log n)$ as long as  $q_j\geq\epsilon n^{-1}$, $j=1,\dots,n$, for some $\epsilon>0$, i.e., if $\Hh{X_j}=\Omega(\log n/n)$. On the other hand, if $q_j\leq q(n)$, $j=1,\dots,n$, for some $q(n)=o(n^{-1})$, then $\Hh{X_j}=o(n^{-1}\log n)$, and thus $\Hh{X}=o(\log n)$.

If all $q_j=1/2$ we get $\Hh{X}\geq\log(n+1)$ by taking $t\to 0$. This is tight for infinitely many values of $n$ (see \sref{s_up}) and confirms \fakeconjref{1.2} of Babai~\cite{B13_En}. However, the following theorem implies the same bound even for the min-entropy and the proof is, in fact, more direct.

\theo[t_Hmin]{Let $X=\l(X_1\dc X_n\r)$ be $n$ pairwise independent binary variables.
Let $q_j=\PR{X_j=1}$ and suppose that $0<q_j<1$ for $j=1\dc n$.
Then
\[
\Hh[min]X\ge\log\l(1+\sum_{j=1}^n
\min\left\{\frac{1-q_j}{q_j},\frac{q_j}{1-q_j}\right\}\r).
\]
}

\prf{Let $U$ be an $m\times (n+1)$ matrix as in the proof of \theoref{t_low}, assuming again w.l.o.g. that $\PR{X_j=1}\le\PR{X_j=0}$ always.
From \bref{m_l1} we get that for all $i\in[m]$,
\m{1\ge p_i\tm\l(1+\sum_{j:a_{ij}=0}\fr{q_j}{1-q_j}+\sum_{j:a_{ij}=1}\fr{1-q_j}{q_j}\r)
\geq p_i\tm\l(1+\sum_{j=1}^n
\min\left\{\frac{1-q_j}{q_j},\frac{q_j}{1-q_j}\right\}\r),}
which gives us the required lower bound on $p_i$'s.
}

\begin{corollary}
If all $q_j\geq q$, then 
\begin{equation}\label{e5}
\Hh[min]X\ge\log\l(1+\fr{nq}{1-q}\r).
\end{equation}
\end{corollary}
For $q=1/2$ (unbiased \pl[X_i]), this corollary gives
\m{\Hh[min]X\ge\log\l(n+1\r),}
which is tight for infinitely many values of $n$. 

\ssect{Pairwise independent finite-outcome variables}

Let $[k]$ be the values that a random variable $X_j$ takes on, $k\geq 2$.

\theo[t_fin]{Let $X=\l(X_1\dc X_n\r)$ be pairwise independent variables that take on values in $[k]$, $k\geq 2$. Let $w$ be such that for all $i\in[n],j\in[k]$,
\m{\PR{X_i=j}\le w
\txt{\tb (i.e., $\Hh[min]{X_i}\ge -\log w$)}.}
If $w\geq\dr12$, then
\m{\Hh[min]X\ge\log\l(\fr{1-w}{w}\tm n+1\r).}
If $w\le\dr12$, then
\m{\Hh[min]X\ge\log(n+1).}
}

To prove the theorem, we need the following technical statement.

\clm[c_com]{For $k\ge2$, let $b_1\dc b_k\ge0$ be such that $\sum_{t=2}^kb_t\ge b_1$ and for all $t\geq 2$, $b_t\leq b_1$.
Then there exist $\alpha_2\dc\alpha_k\in\RR$, such that
\m{\sum_{t=2}^t\en^{\iu\alpha_t}b_t=b_1.}}

\prf[\clmref{c_com}]{Let $C_r$ denote the circle in the complex plane with radius $r$ and center in~$0$. The claim is equivalent to the statement that $b_1$ is in the Minkowski sum of $C_{b_2},\dots,C_{b_k}$.
Note that if $r\leq s$, then $C_r+C_s$ contains $C_s$ as a subset.
Thus the sum $C_2+\dots+C_k$ is either a region between $C_{b_2+\dots+b_k}$ and some smaller circle, or a disc with radius $b_2+\dots+b_k$ -- in any case, it contains both $C_{\max_{2\leq t\leq k}b_t}$ and $C_{b_2+\dots+b_k}$.
Hence it also contains $b_1$.
}

\prf[\theoref{t_fin}]{The proof is a modification of the proofs of \theoref[t_low]{t_Hmin}.

Let $A=\set{a_1\dc a_m}\sbseq[k]^n$ be the support of $X$ and $p_i\deq\PR{X=a_i}$.
For $j\in[n]$, let $w_j\deq\Max[{t\in[k]}]{\PR{X_j=t}}$ and assume without loss of generality that $\PR{X_j=1}=w_j$.
Let $\omega_j=\max\{1,\fr{w_j}{1-w_j}\}$, and let $\alpha_{j2}\dc\alpha_{jk}$ be the values guaranteed by \clmref{c_com} for $b_1=w_j/\omega_j$ and $b_t=\PR{X_j=t}$ for $2\le t\le k$ (which observes the claim requirements).

This time we define the matrix $U$ over $\CC$:\ for $i\in[m]$,
\m{u_{i 0}\deq \sqrt{p_i},}
and for $j\in[n]$,
\m{u_{i j}\deq\twocase
{-\sq{\fr[/]{p_i}{\omega_j}}}{if $a_{i j}=1$}
{\sq{p_i\omega_j}\tm \en^{\iu\alpha_{jz}}}{if $a_{i j}=z>1$}~.}

As before, let $u_j$ denote the \ord[j] column vector of $U$.
Then, by the immediate adaptation of the argument we gave for \theoref{t_low} (taking into account the guarantees of \clmref{c_com}), it holds that for all $j\neq k>0$,
\m{\ip{u_{j}}{u_{k}}=0 \txt{~~and~~} \norm{u_{j}}=1.}
Therefore, the norm of each row of $U$ is at most $1$ and, for every $i$,
\m{1\ge p_i\tm\l(1+\fr n{\omega_{max}}\r),}
where $\omega_{max}\deq\Max[{j\in[n]}]{\omega_j}$.
The result follows.
}

\ssect{$d$-wise independent unbiased binary variables}

One can use an idea from~\cite{ABI86_A_Fa} to derive from the case of pairwise independent variables stronger lower bounds for $d$-wise independent variables.
We will demonstrate it only on \theoref{t_Hmin}, but the same idea can be used together with  other lower bounds on the entropy of pairwise independent variables.

\theo[t_lun]{Let $X=\l(X_1\dc X_n\r)$ where \pl[X_j] are \f d-wise independent unbiased binary variables.
If $d$ is even, then 
\m{\Hh[min]X\geq
\log\l(\sum_{i=0}^{{d/2}}\chs ni\r).}
If $d$ is odd, then
\m{\Hh[min]X
\ge\log\l(\sum_{i=0}^{(d-1)/2}\chs{n}{i} + \chs{n-1}{(d-1)/2}\r).}
}

\prf{

Let $d$ be even.
We define $Y=\l(Y_1\dc Y_m\r)$, where all \pl[Y_i] are unbiased binary variables equal to the parity of at most ${d/2}$ variables $X_i$ and $m=\sum_{i=1}^{{d/2}}\chs ni$ (every $Y_i$ is unique).
Clearly, $Y_1\dc Y_m$ are pairwise independent, and from \theoref{t_Hmin} we get
\m{\Hh[min]X\ge\Hh[min]Y\ge
\log\left(1+\sum_{i=1}^{{d/2}}\chs ni\right)
= \log\l(\sum_{i=0}^{{d/2}}\chs ni\r).}

If $d$ is odd, we take the parities of at most $(d-1)/2$ variables $X_i$ and the parities of $X_1$ with exactly $(d-1)/2$ other variables.
The resulting variables are again pairwise independent.

}

In the next section we will see, in particular, that the above bound is tight for the case of $d=3$ and $n$ being a power of $2$. 

\sect[s_up]{Upper bounds}

In this section we review some constructions of $d$-wise independent
unbiased binary variables with low entropies.  The constructions
are based on known ideas, and they are included here to argue
optimality of the lower bounds from \sref{s_low}. 

The standard way of constructing \f d-wise independent distributions
is using parity check matrices of codes with minimum
distance $\geq d$.  In such matrices every $d$ columns are linearly
independent.  Hence, if we take the space of vectors generated by the
rows of such a matrix, i.e., the dual code, we obtain $d$-wise
independent variables.
Over $GF_2$, these are balanced binary variables.
To get matching bounds we have to find suitable codes. 

We start with the case of pairwise independent variables ($d=2$).
Recall that an Hadamard matrix is a real matrix with entries $\pm 1$
whose rows (and hence also columns) are orthogonal.  Hadamard matrices
exist for infinitely many dimensions, in particular for every power
of~$2$. Given an Hadamard matrix of dimension $n+1$, first transform
it into an Hadamard matrix with the first column having all 1s, then
delete the first column. The resulting $(n+1)\times n$ matrix defines
in a natural way an \emph{Hadamard code} and $n$ pairwise independent
balanced binary variables supported on a set of size $n+1$.

Lancaster~\cite{L65_Pair} proved:
\begin{enumerate}
\item For every $n\geq 2$, there exist at most $n$ pairwise independent random variables on a probability space with $n+1$ points.
\item The existence of such random variables where, additionally, each point in the
probability space has measure $\frac 1{n+1}$ is equivalent to the
existence of an Hadamard matrix of dimension $n+1$.
\end{enumerate}
Our proofs of Theorems~\ref{t_low} and~\ref{t_Hmin} can be viewed as an
extension of an argument used by Lancaster to prove 2.
Lancaster considered general (finite-outcome) pairwise independent variables.
For unbiased binary variables, we can prove the following.

\begin{theorem}
The existence of $n$ pairwise independent unbiased binary
variables with entropy equal to $\log(n+1)$ is equivalent to the
existence of an Hadamard matrix of dimension $n+1$.
\end{theorem}

\begin{proof}
As shown above, an Hadamard matrix of dimension $n+1$ gives rise to
$n$ pairwise independent unbiased binary variables with entropy
equal to $\log(n+1)$.

To prove the converse, let $n$ pairwise independent unbiased
binary variables with entropy equal to $\log(n+1)$ be
given. According to Theorem~\ref{t_Hmin}, every point in the
probability space has measure at most  $\frac 1{n+1}$. Since the
entropy is $\log(n+1)$, this implies that there are exactly $n+1$ points,
each with measure $\frac 1{n+1}$. The existence of an Hadamard
matrix of dimension $n+1$ then follows from Lancaster's theorem, or from
our proof of Theorem~\ref{t_low}. 
\end{proof}

\medskip
Another case where we can precisely match the lower bound for infinitely many
values of $n$ is $d=3$. Let $n=2^l$ and consider the $(l+1)\times
n$ binary matrix whose first row consists of \pl[1] and the columns
restricted to the remaining $l$ rows are all vectors of length $l$.
Every two different columns are linearly independent over $\GF[l+1]2$ because they are different.
Every three different columns are also independent because
every two of them are and they cannot sum to zero vector due to the
first row.  Hence the space generated by the rows is $3$-wise
independent.  The size of the space is $2^{l+1}=2n$, precisely
matching the statement of \theoref{t_lun}. Thus we have:

\begin{theorem}
If $n$ is a power of $2$, then the minimum of $\Hh[min]X$ taken over all $n$-tuples of $3$-wise independent unbiased binary variables is $\log 2n$. 
\end{theorem}

Note that the above construction is based on the parity-check matrix of the Hamming code: first we extend the matrix by a column with all zeros and then we extend it by a row with all ones. The two constructions, one based on the Hadamard code and the other based on the Hamming code, can be generalized using BCH codes. Recall that the binary BCH code of length $2^m-1$ and designed distance $2t+1$ has the minimal distance at least $2t+1$ and dimension $2^m-1-mt$, provided that $m$ is sufficiently large with respect to $t$ (see~\cite{macwilliams-sloane}, pages 258 and 253). Hence every $2t$ columns of the parity-check matrix (and also of the dual code) are linearly independent and the dimension of the space generated by the parity-check matrix (i.e., the dual code) has dimension $mt$. Thus for $d>2$ even, we can take a BCH code with designed distance $2t+1=d+1$ and we get $n=2^m-1$ $d$-wise independent random variables with min-entropy
\[
\frac d2\log(n+1).
\]
For $d>3$ odd, we take $2t+1=d$, and extend the parity-check matrix matrix by a column of zeros and a row of ones, as we did above. Thus we obtain a matrix with every $d$ columns independent. Let $n=2^m$ be the number of columns of this matrix. The linear space generated by the rows gives a probability space of $n$ $d$-wise independent random variables with min-entropy
\[
\frac{d-1}2\log n+1.
\]
These bounds are asymptotically equal to 
the lower bounds of \theoref{t_lun} when $n$ goes to infinity. However, we have
not been able to find constructions matching our lower bound exactly
for any $d\geq 4$ and any $n$.

\section{Conclusions}

We proved several lower bounds on the entropy of pairwise and $d$-wise independent random variables. Our lower bounds match upper bounds exactly, or asymptotically for some special values of the parameters involved. But for most values of parameters, we do not know even the asymptotic behavior of the dependence of entropy on them. This is, in particular, so in the case of equally distributed pairwise independent 0-1 variables. In this special case we have two bounds~(\ref{e4}) and~(\ref{e5}), which give an asymptotically optimal bound for $q\approx 1/n$ and a tight bound for $q=1/2$, but for other values we do not know. Another interesting problem, studied in~\cite{B13_En}, is to find the best lower bound on the joint entropy $\Hh{X_1,\dots,X_n}$ of a string of pairwise random variables $ X_1,\dots,X_n$ in terms of the parameter $L:=\sum_j X_j$. 
For more open problems, see~\cite{B13_En}.

\subsection*{Acknowledgment}
We would like to thank an anonymous referee for suggesting better formulations of some theorems and an elegant proof of Claim~\ref{c_com}.

\end{document}